\documentclass{article}
\usepackage{graphicx}
\usepackage{amsmath}
\usepackage{amssymb}
\usepackage{amsthm}
\usepackage{bbm}

\usepackage[maxbibnames=20,sorting=none]{biblatex}
\usepackage{authblk}

\newtheorem{thm}{Theorem}
\newtheorem{prop}{Proposition}

\newtheorem{lem}{Lemma}

\title{\Large Cryptanalysis of Semidirect Product Key Exchange Using Matrices Over Non-Commutative Rings}
\author{\normalsize Christopher Battarbee, Delaram Kahrobaei, and Siamak F.\ Shahandashti}
\affil{Department of Computer Science, University of York, UK \\
\texttt{\{cb203, dk928, siamak.shahandashti\} @york.ac.uk}}

\begin{document}

\maketitle

\begin{abstract}
    It was recently demonstrated that the Matrix Action Key Exchange (MAKE) algorithm, a new type of key exchange protocol using the semidirect product of matrix groups, is vulnerable to a linear algebraic attack if the matrices are over a commutative ring. In this note, we establish conditions under which protocols using matrices over a non-commutative ring are also vulnerable to this attack. We then demonstrate that group rings $R[G]$, where $R$ is a commutative ring and $G$ is a non-abelian group, are examples of non-commutative rings that satisfy these conditions.
\end{abstract}

\section{Introduction}
Since the advent of Shor's algorithm, it has been desirable to study alternatives to the Diffie-Hellman key exchange~\cite{diffie1976new}. One approach to this problem appeals to a more complex group structure: recall that for (semi)groups $G,H$ and a homomorphism $\theta:H\to Aut(G)$, the semidirect product of $G$ by $H$ with respect to $\theta$, $G\rtimes_{\theta}H$, is the set of ordered pairs $G\times H$ equipped with multiplication 
\[(g,h)(g',h')=(\theta(h')(g)g',hh')\]
Recall also that the action of a group $G$ on a finite set $X$ is a function $(G,X)\to X$, here written as $g\cdot x$, satisfying $1\cdot x=x$ and $g\cdot(h\cdot x)=(gh)\cdot x$ for all $g,h\in G$. It turns out that such an action induces a homomorphism into the group of permutations of $X$; in particular, if $G,H$ are groups, an action of $H$ on $G$ specifies a homomorphism into the automorphism group of $G$, so specifying such an action suffices to specify a semidirect product structure.

The semidirect product can be used to generalise the Diffie-Hellman key exchange \cite{habeeb2013public} via a general protocol sometimes known as the ``non-commutative shift''. Originally, the semigroup of $3\times 3$ matrices over the group ring $\mathbb{Z}_7[A_5]$ is proposed as the platform; however, this turned out to be vulnerable to the type of attack (the so-called ``dimension attack'') by linear algebra described in \cite{myasnikov2015linear},\cite{roman2015linear}. Other platforms used include tropical algebras \cite{grigoriev2014tropical} and free nilpotent $p$-groups \cite{kahrobaei2016using}. The former is shown to be insecure in \cite{isaac2021closer}, \cite{kotov2018analysis}.

The insight of the recent MAKE protocol \cite{rahman2020make} is to use the ring formed by square matrices over a ring. This object is a group under addition and a semigroup under multiplication, so we can follow the syntax of \cite{habeeb2013public} in such a way as to mix operations so that no power of any matrix is ever exposed. However, the protocol is vulnerable to another linear algebraic attack \cite{brown2021cryptanalysis}, which relies on the commutativity of the underlying ring. The purpose of this note is to demonstrate that under certain circumstances, using a non-commutative underlying ring will have the same vulnerability. In particular, we present general conditions by which one can decide if a platform to be used with MAKE is unsafe. It turns out these conditions are satisfied by group rings of the form used in \cite{habeeb2013public}; note that we do not claim to present a break of \cite{habeeb2013public} via our methods.

\section{Matrix Action Key Exchange (MAKE)}
The following is taken from \cite{rahman2020make}, following an original version in which $H_1=H_2$.

For $n\in\mathbb{N}$ and $p$ prime, consider the additive group $G$ of $n\times n$ matrices over $\mathbb{Z}_p$, $M_n(\mathbb{Z}_p)$, and the semigroup $S=\{(H_1^i,H_2^i):i\in\mathbb{N}\}$ generated by non-invertible matrices $H_1,H_2\in M_n(\mathbb{Z}_p)$. The action of $S$ on $G$ defined by $(H_1^i,H_2^i)\cdot M=H_1^iMH_2^i$\footnote{We rely on commutativity of $S$ to satisfy the axioms of an action, which is why a cyclic (semi)group is used.} induces a homomorphism into the automorphism group of $G$; we can therefore define the semidirect product of $G$ by $S$ with multiplication
\[(M,(H_1^i,H_2^i))(M',(H_1^j,H_2^j))=(H_1^jMH_2^j+M',(H_1^{i+j},H_2^{i+j}))\]
In particular one checks that for any choice of $H_1,H_2$, exponentiation has the form 
\[(M,(H_1,H_2))^n=\left(\sum_{i=0}^{n-1}H_1^iMH_2^i,(H_1^n,H_2^n)\right)\]
We use this semidirect product structure in the syntax of  \cite{habeeb2013public} as follows. Suppose Alice and Bob wish to agree on a shared, private key by communicating over an insecure channel. Suppose also that public data $M,H_1,H_2$ is available.

\begin{enumerate}
    \item Alice picks random $x\in\mathbb{N}$ and calculates $(M,(H_1,H_2))^x=(A,(H_1^x,H_2^x))$ and sends $A$ to Bob.
    \item Bob similarly calculates a value $B$ corresponding to random $y\in\mathbb{N}$, and sends it to Alice.
    \item Alice calculates $(B,*)(A,(H_1^x,H_2^x))=(H_1^xBH_2^x+A,**)$ and arrives at her key $K_A=H_1^xBH_2^x+A$. She does not actually calculate the product explicitly since she does not know the value of $*$; however, it is not required to calculate the first component of the product.
    \item Bob similarly calculates his key as $K_B=H_1^yAH_2^y+B$.
\end{enumerate}

Since $A=\sum_{i=0}^{x-1}H_1^iMH_2^i$, $B=\sum_{i=0}^{y-1}H_1^iMH_2^i$, we have
    \begin{align*}
        H_1^xBH_2^x+A &= H_1^x\left(\sum_{i=0}^{y-1}H_1^iMH_2^i\right)H_2^x+A \\
        &= \sum_{i=x}^{x+y-1}H_1^iMH_2^i+\sum_{i=0}^{x-1}H_1^iMH_2^i \\
        &= \sum_{i=y}^{x+y-1}H_1^iMH_2^i+\sum_{i=0}^{y-1}H_1^iMH_2^i \\
        &= H_1^yAH_2^y+B
    \end{align*}
Alice and Bob therefore both arrive at the same shared key $K=K_A=K_B$.

Attacking the protocol directly requires recovering $x,y$ from $A,B$. This leads to a natural analogue of the computational Diffie-Hellman assumption; namely, computational infeasibility of retrieving the shared secret $K$ given the data $(H_1,H_2,M,A,B)$\footnote{This is a weaker security notion than key indistinguishability, analogue of the decisional Diffie-Hellman assumption; the authors of \cite{rahman2020make} conduct some computational experiments suggesting the latter assumption may hold. This fact is not further referenced in this paper, since the attack does not require solving the analogue of the discrete log problem.}. Clearly, this is closely related to an analogue of the discrete logarithm problem (DLP), which is shown in \cite{rahman2020make} to be at least as hard as the standard DLP provided certain ``safe'' primes $p$ are used.

\section{Attack by Cayley-Hamilton}
Several protocols following the non-commutative shift syntax are vulnerable to the dimension attack, which does not require one to solve the problems addressed in the security assumption. This class of attacks, however, deal with schemes using only group multiplication. In our case, we have two operations; the following attack was developed by Brown, Koblitz and Legrow in \cite{brown2021cryptanalysis} and is roughly outlined below. Suppose the public data $M,H_1,H_2$ are fixed, as well as transmitted values $A,B$ corresponding to exponents $x,y$ respectively.

The attack relies on the following easily-verifiable fact: we have that
\[H_1AH_2 + M - A = H_1^xMH_2^x\]
This identity is known as the ``telescoping'' equality. It is crucial to allow the recovery of the quantity $H_1^xMH_2^x$ from the data available to an eavesdropper on the left-hand side of the equality.

Suppose the matrices are of size $n\in\mathbb{N}$. We also rely on the Cayley-Hamilton theorem, which for a square matrix $A$ over $M_n(\mathbb{Z}_p)$ and any $x\in\mathbb{N}$ guarantees the existence of coefficients $p_i$ in $\mathbb{Z}_p$ such that 
\[A^x = \sum_{i=0}^{n-1}p_i A^i\]

Finally, we need the following two-part lemma, the proof of which is given in \cite{brown2021cryptanalysis}.

\begin{lem}\label{lem:attack-lemmas}
Let $n\in\mathbb{N}$. Define $L:M_n(\mathbb{Z}_p)\to M_{n^2}(\mathbb{Z}_p)$ component-wise by
\[(L(Y))_{jn+i,hn+g}=(H_1^gYH_2^h)_{i,j}\]
for $0\leq i,j,g,h\leq n-1$, and $vec:M_n(\mathbb{Z}_p)\to\mathbb{Z}_p^{n^2}$ by
\[vec(A)_{jn+i} = A_{i,j}\]
for $0\leq i,j \leq n-1$. Then there is a vector $s$ in $\mathbb{Z}_p^{n^2}$ such that $L(Y)s=vec(H_1^xYH_2^x)$ for any $Y\in M_n(\mathbb{Z}_p)$. Moreover, for some $Y\in M_n(\mathbb{Z}_p)$, a vector $u\in\mathbb{Z}_p^{n^2}$ satisfying $L(Y)u=0$ also satisfies $L(H_1^lYH_2^l)u=0$ for any $l\in\mathbb{N}$.
\end{lem} 

The attack now works as follows:
\begin{enumerate}
    \item Using the telescoping equality, recover the value $H_1^xMH_2^x$.
    \item Solve the $n^2$ linear equations in $n^2$ unknowns defined by $L(M)t=vec(H_1^xMH_2^x)$ to recover a vector $t$; by Lemma~\ref{lem:attack-lemmas}, there is at least one solution to this system of equations; and any solution satisfies $L(B)t=vec(H_1^xBH_2^x)$.
    \item Since $vec$ is a bijection, applying its inverse to $L(B)t$ allows one to recover $H_1^xBH_2^x$, and therefore the shared key $K$ by simply adding $A$ to this quantity.
\end{enumerate}

\section{Attacking Non-Commutative Rings}

A key part of the above attack is the construction of the vector $s$, which is done by the Cayley-Hamilton theorem. In particular, this theorem only applies to square matrices over commutative rings; we will use the following theorem to characterise some non-commutative rings over which the scheme is still insecure. 
In the following, let $R$ be an arbitrary non-commutative ring.

\begin{thm}\label{thm:inj-ring-hom}
Suppose there is an injective ring homomorphism $\phi:R\to M_m(S)$ for some $m\in\mathbb{N}$ and a commutative ring $S$. For any $n\in\mathbb{N}$ define 
\begin{align*}
    \psi: &M_n(R)\to M_{mn}(S) \\
    &(\psi(A))_{im+g,jm+h}=(\phi(A_{i,j}))_{g,h}
\end{align*}
where $0\leq i,j\leq n-1$, $0\leq g,h \leq m-1$. Then $\psi$ is an injective ring homomorphism.
\end{thm}

\begin{proof}
To check multiplication is preserved we just check that the relevant quantities agree on each entry. Let $A, B$ in $M_n(R)$; then
\begin{align*}
    (\psi(AB))_{im+g,jm+h} &= (\phi((AB)_{i,j}))_{g,h} \\
    &= \left(\phi\left(\sum_{k=0}^{n-1}A_{i,k}B_{k,j}\right)\right)_{g,h} \\
    &= \left(\sum_{k=0}^{n-1}\phi(A_{i,k})\phi(B_{k,j})\right)_{g,h} \\
    &= \sum_{k=0}^{n-1}(\phi(A_{i,k})\phi(B_{k,j}))_{g,h} \\
    &= \sum_{k=0}^{n-1}\sum_{l=0}^{m-1}\phi(A_{i,k})_{g,l}\phi(B_{k,j})_{l,h} \\
    &= \sum_{k=0}^{n-1}\sum_{l=0}^{m-1}\psi(A)_{im+g,km+l}\psi(B)_{km+l,jm+h} \\
    &= (\psi(A)\psi(B))_{im+g,jm+h}
\end{align*}

Similarly, for addition, we have 
\begin{align*}
    (\psi(A+B))_{in+g,jn+h} &= (\phi((A+B)_{i,j}))_{g,h} \\
    &= (\phi(A_{i,j})+\phi(B_{i,j}))_{g,h} \\
    &= (\phi(A_{i,j}))_{g,h}+(\phi(A_{i,j}))_{g,h}
\end{align*}
Finally, $\psi(I_n)=I_{mn}$ since $\phi(1)=I_m$, so $\psi$ is a ring homomorphism. To see injectivity, for $A,B\in M_n(R)$ suppose $\psi(A)=\psi(B)$. Then for each $0\leq i,j\leq n-1$, $0\leq g,h \leq m-1$ we have $\phi(A_{i,j})_{g,h}=\phi(B_{i,j})_{g,h}$. Therefore $\phi(A_{i,j})=\phi(B_{i,j})$ for each $i,j$. Since $\phi$ is injective, we must have $A=B$.
\end{proof}

Once we have established that $\psi$ is indeed a ring homomorphism the attack can just be carried out on $\psi$ applied to the public matrices. The details are listed below for completeness.

\subsection{Extending the Attack}

Letting $k=mn$ we have a function $L\circ\psi:M_n(R)\to M_{k^2}(S)$ defined by 
\[(L(\psi(Y)))_{jk+i,hk+g}=(\psi(H_1^gYH_2^h))_{i,j}\]
where each of the indices run from $0$ to $k-1$. The function $vec$ (defined with a different domain in Lemma~\ref{lem:attack-lemmas}) stacks the columns of a matrix in $M_k(S)$ to give a column vector of height $k^2$. 

We will need to invoke the following two propositions during the attack:

\begin{prop}\label{prop:s-exists}
There is a vector $s\in S^{k^2}$ such that for all $Y\in M_n(R)$, we have 
\[L(\psi(Y))s=vec(\psi(H_1^xYH_2^x))\]
\end{prop}

\begin{prop}\label{prop:zero}
Suppose some vector $u$ is such that $L(\psi(Y))u=0$ for $Y\in M_n(R)$. Then for all $l\in\mathbb{N}$ we have $L(\psi(H_1^lYH_2^l))u=0$.
\end{prop}

The proofs are somewhat tedious and similar to those given in \cite{brown2021cryptanalysis}; the interested reader can find them in the appendix.

For the public parameters $H_1,H_2,M$ and fixed values of $A,B$  we can calculate
\[\psi(M+H_1AH_2-A)=\psi(H_1^xMH_2^x)\]

By Proposition~\ref{prop:s-exists}, the equation 
\[L(\psi(M))t=vec(\psi(H_1^xMH_2^x))\]
has at least one solution. We can therefore solve this system of linear equations efficiently, for example by Gaussian elimination, and obtain a solution, say $t$. We know that, with $Y=B$, we also have \[L(\psi(B))s=vec(\psi(H_1^xBH_2^x))\]

Since the vectors $t$ and $s$ satisfy $L(\psi(M))t=L(\psi(M))s$ and $L$ preserves addition, setting $u=t-s$ we have, invoking Proposition~\ref{prop:zero}, that 

\begin{align*}
    0 &= L(\psi(M))u+L(\psi(H_1MH_2))u+...+(L(\psi(H_1^{y-1}MH_2^{y-1}))u \\
    &= L(\psi(M)+\psi(H_1MH_2)+...+\psi(H_1^{y-1}MH_2^{y-1}))u \\
    &= L(\psi(M+H_1MH_2+...+H_1^{y-1}MH_2^{y-1}))u \\
    &= L(\psi(B))u
\end{align*}

Therefore $L(\psi(B))t=L(\psi(B))s=vec(\psi(H_1^xBH_2^x))$, so from public information we can recover $\psi(H_1^xBH_2^x)$, and hence 
\begin{align*}
    \psi(K)&= \psi(A+H_1^xBH_2^x) \\
    &= \psi(A)+\psi(H_1^xBH_2^x)
\end{align*}

Note that the vector $s$ is not available from public information, but at no point is its calculation required. It is merely described to show that the vector $t$ recovered by the attacker will indeed suffice for recovery of $\psi(K)$.

In general, recovering $K$ from $\psi(K)$ can be done by inverting $\phi$ on the $n^2$ blocks of size $m\times m$ of $\psi(K)$. This is trivial if there is an explicit description of $\phi$.

\section{Group Ring Representations}

A well-behaved and easily scalable example of non-commutative rings are group rings of the form $R[G]$, where $R$ is a commutative ring and $G$ is a non-abelian group. For example, $\mathbb{Z}_7[A_5]$ is used in \cite{habeeb2013public}. We now show that such a ring meets the conditions required for the above modification of the attack. The following definitions are taken from \cite{milies2002introduction}, to which the reader is referred for more detail.

Let $G$ be a finite group, $R$ be a ring. Consider the set of formal sums

\[R[G]=\left\{\sum_{g\in G}a_g.g:a_g\in R,g\in G\right\}\]
where the multiplication refers to scalar multiplication\footnote{Technically speaking, the formal sums refer to linear combinations of functions from $G$ to $R$. However, once we have defined such functions we usually dispense with them in favour of the notation above; see \cite{milies2002introduction} for further details.}. Together with addition and multiplication defined respectively by
\[\sum_{g\in G}a_g. g+\sum_{g\in G}b_g. g=\sum_{g\in G}(a_g+b_g). g \quad \left(\sum_{g\in G}a_g.g\right)\left(\sum_{h\in G}b_h. h\right)=\sum_{g,h\in G}(a_gb_h). gh\]
$R[G]$ is a ring that is at the same time a free left $R$-module with basis $G$. Moreover, $G$ acts on $R[G]$ by left multiplication:

\[g\cdot\sum_{h\in G}a_h. h=g\sum_{h\in G}a_h. h=\sum_{h\in G}a_h. (gh)\]

Suppose $|G|=m$. Note that left multiplication by a group element permutes the group, which is the basis of $R[G]$, the $R$-module of rank $m$. As a function, then, this multiplication is an automorphism of the $R$-module; there is therefore a function $T:G\to GL(k,R)$, where the function $T(g)$ has matrix representation with entries in $R$. This is the so-called ``left-regular representation'' of $G$ over $R$. Moreover, one can easily verify that this map is a group homomorphism.

The matrix representation of the function $T(g)\in GL(k,R)$ is not unique and depends on a choice of basis. However, since the group $G$ is a basis of $R[G]$, and $T(g)$ permutes this basis, we can specify the matrices as follows. Enumerate the elements of $G$ arbitrarily, and write $T_{g_i}$ for the matrix corresponding to the function $T(g_i)$. Suppose $g_ig_j=g_k$, then $(T_{g_i})_{k,j}=1$, with all other entries in the row 0. In this way we can construct a set of matrices $\{T_g:g\in G\}$ from a multiplication table of $G$.

\subsection{Mapping to Matrices over a Commutative Ring}

We can extend the left-regular representation outlined above to a map

\[\phi:R[G]\to M_m(R):\sum_{g\in G}a_g.g\mapsto \sum_{g\in G}a_g. T_g\]

Note that the sum of scaled invertible matrices is not necessarily invertible; hence, the map is into $M_m(R)$, rather than $GL(m,R)$.

\begin{prop}\label{prop:phi-exists}
Suppose $R$ is a commutative ring. We have that $\phi:R[G]\to M_m(R)$ is an injective ring homomorphism.
\end{prop}

\begin{proof}
Clearly $\phi$ is an additive homomorphism. To show multiplication is preserved note that since $R$ is commutative we have 

\[\sum_{g,h\in G}(a_g b_h).T_{gh}=\sum_{g,h\in G}(a_g b_h).T_g T_h=\sum_{g\in G}a_g.T_g \sum_{h\in G}b_h.T_h\]

Preservation of the identity is inherited from the homomorphicity of the map $T$. To see that $\phi$ is injective, we first show that $\phi$ is injective exactly when the matrices $\{T_g:g\in G\}$ are linearly independent over $R$. This is because $\ker\phi=\{0\}$ exactly when the only coefficients $a_g$ that give $\sum_{g\in G}a_g.T_g=0$ are all zero, i.e. when the matrices are linearly independent, and the kernel is trivial if and only if the map is injective.
Suppose for contradiction that matrices $T_{g_i},T_{g_j}$ have a 1 in the same place, say the $m,n$th entry. By the construction of such matrices given above, this means that for $g_i\neq g_j$ we have $g_ig_m=g_n=g_jg_m$, which is a contradiction, since the action of a group on itself by left multiplication is faithful. Clearly, this implies the matrices are linearly independent, and so $\phi$ is injective.
\end{proof}

We therefore have the required homomorphism $\phi$, from which $\psi$ can be constructed as in the general case. 

\subsection{Inverting $\psi$}
We can recover the unique value of $K$ as follows. The $mn\times mn$ matrix $\psi(K)$ recovered in the above consists of $n^2$ blocks of size $m\times m$, where the $i,j$th block is given by $\phi(K_{i,j})$. We know from the proof of Proposition~\ref{prop:phi-exists} that the matrices $T_g$ are a basis of the image of $\phi$, so $\phi(K_{i,j})$ has unique decomposition as $\phi(K_{i,j})=\sum_{g\in G}k_{g,i,j}.T_g$. Given the values of $T_g$, finding this decomposition amounts to solving $m$ linear equations in $m$ unknowns. By definition of $\phi$ we have $K_{i,j}=\sum_{g\in G}k_{g,i,j}.g$; repeating this procedure for each $i,j$, we recover $K$ from $\psi(K)$ in polynomial time.

\section{Conclusions}
We again stress that the attack described in this paper effectively bypasses the security assumption made in \cite{rahman2020make}. As remarked in \cite{brown2021cryptanalysis} this is another example of some inherent linearity underpinning matrix-based key exchange protocols.

The main limiting factor in the efficiency of this attack is recovering the vector $t$ by solving $(mn)^2$ linear equations in $(mn)^2$ unknowns. Since solving $n$ linear equations in $n$ unknowns has a complexity\footnote{Asymptotically faster methods are available but are outside the scope of this paper, and may only represent an overall increase in efficiency in the case that very large matrices are used.} of $\mathcal{O}(n^3)$, we expect the time complexity of the attack to be $\mathcal{O}((mn)^6)$. Should one wish to use a ring $R$ satisfying the conditions of Theorem~\ref{thm:inj-ring-hom}, therefore, one should ensure that $m$ is large, where $\phi:R\to M_m(S)$, and $S$ is a commutative ring. For sufficiently large values of $m$ the attack becomes infeasible, although the complexity is still polynomial.

In the case of group rings $R[G]$ this is possible to achieve by increasing the size of the group $G$. However, we constructed $\phi$ from the left regular representation of $G$ over $R$, where the dimension of the representation and therefore $m$ is always the size of $G$. For some groups it might be possible to construct $\phi$ from a faithful representation of lower dimension, so one should use a group where there is a lower bound on the dimension of a faithful representation; for example, certain $p$-groups \cite{janusz1971faithful}. This fact was used to counter similar attacks in \cite{kahrobaei2016using}.

It is an interesting problem to determine for which non-commutative rings there is no injective homomorphism into matrices over a commutative ring; such rings would be safe from the attack of \cite{brown2021cryptanalysis}, and the attack could not be extended by the methods described in this paper. In some sense, then, the criteria described in Theorem~\ref{thm:inj-ring-hom} serve to classify rings into ``safe'' or ``unsafe'' for use with the MAKE protocol. 

Finally, we note that although the group rings used in \cite{habeeb2013public} satisfy the conditions of Theorem~\ref{thm:inj-ring-hom}, our method does not present a break of the scheme in \cite{habeeb2013public}. This is effectively because the exchanged values $A,B$ are calculated as product, rather than a sum, and the function $L$ does not preserve multiplication. Moreover, whilst there is an analogue of the telescoping equality in that context, it does not necessarily allow recovery of the required quantity because the exchanged values do not always have a multiplicative inverse (in contrast to the values exchanged during the MAKE protocol, which always have additive inverse). On the other hand, Theorem~\ref{thm:inj-ring-hom} does give us access to the Cayley-Hamilton theorem in the context of \cite{habeeb2013public}.

\section{Acknowledgements}
We thank Chris Monico for his helpful correspondence on his paper and Vladimir Shpilrain for correspondence on the manuscript, as well as Alfred Dabson at City, University of London for advice on various technical details. We thank the anonymous referees for their helpful suggestions.

\newpage

\printbibliography

\newpage
 
\section{Appendix} 
Here we detail the proofs of Propositions \ref{prop:s-exists} and \ref{prop:zero}.

\begin{proof}[Proof of Proposition \ref{prop:s-exists}]
Since we are now working with matrices over a commutative ring, by the Cayley-Hamilton theorem we have $p_i,q_i\in S$ such that 
\[\psi(H_1)^x=\sum_{g=0}^{k-1}p_g\psi(H_1)^g \quad \psi(H_2)^x=\sum_{h=0}^{k-1}q_h\psi(H_2)^h\]

With $T\in M_{k}(S)$ defined by $T_{i,j}=p_iq_j$ and $s=vec(T)$ we have, for any $Y$ in $M_k(S)$, that 
\begin{align*}
    (L(\psi(Y))s)_{jk+i} &= \sum_{g,h=0}^{k-1}\left(\psi(H_1^gYH_2^h\right))_{i,j}p_gq_h \\
    &= \sum_{g,h=0}^{k-1}(p_g\psi(H_1)^x\psi(Y)q_h\psi(H_2)^h)_{i,j} \\
    &= (\psi(H_1)^x\psi(Y)\psi(H_2)^x)_{i,j} \\
    &= vec(\psi(H_1^xYH_2^x))_{jk+i}
\end{align*}

Therefore $L(\psi(Y))s=vec(\psi(H_1^xYH_2^x))$. \end{proof}

\begin{proof}[Proof of Proposition \ref{prop:zero}]
Checking component-wise, from the definitions it follows that 
\[L(\psi(H_1^lYH_2^l))u=vec\left(\sum_{g,h=0}^{k-1}(\psi(H_1)^g\psi(H_1^lYH_2^l)\psi(H_2)^h)u_{hn+g}\right)\]
and
\[\sum_{g,h=0}^{k-1}\psi(H_1^gYH_2^h)u_{hn+g}=vec^{-1}(L(\psi(Y))u)\]

Therefore, using that $\psi$ preserves multiplication, we have 
\begin{align*}
    L(\psi(H_1^lYH_2^l))u &= vec\left(\sum_{g,h=0}^{k-1}(\psi(H_1)^g\psi(H_1^lYH_2^l)\psi(H_2)^h)u_{hn+g}\right) \\
    &= vec\left(\psi(H_1)^l\left(\sum_{g,h=0}^{k-1}\psi(H_1^gYH_2^h)u_{hn+g}\right)\psi(H_2)^l\right) \\
    &= vec(\psi(H_1)^lvec^{-1}(L(\psi(Y))u)\psi(H_2)^l) \\
    &= vec(0)=0.
\end{align*}
since clearly $vec(0)$ is the zero vector height $k^2$, and $vec$ is a bijection. 
\end{proof}

\newpage

\end{document}